\documentclass[12pt,tightenlines,eqsecnum,floats,showpacs,nofootinbib, aps, prd]{revtex4-2}
\usepackage{enumitem}
\usepackage{graphicx}
\usepackage{subcaption}
\usepackage{amsmath,amsthm}
\usepackage{amsfonts}
\usepackage{amssymb}
\usepackage{mathrsfs}
\usepackage[dvipsnames]{xcolor}
\usepackage{epstopdf}
\usepackage[utf8]{inputenc}
\usepackage{natbib}
\usepackage[colorlinks,urlcolor=NavyBlue,citecolor=NavyBlue,linkcolor=NavyBlue,pdfusetitle]{hyperref}
\usepackage[all]{hypcap}
\usepackage{orcidlink}

\newtheorem{theorem}{Theorem}
\newtheorem{thm}{Theorem}

\newtheorem{lemma}{Lemma}
\newtheorem{proposition}{Proposition}
\newtheorem{corollary}{Corollary}

\newtheorem{definition}{Definition}

\theoremstyle{remark}
\newtheorem{remark}{Remark}

\newcommand{\dd}{\,d}

\begin{document}

\title{Scalar charge bounds for extremal black hole formation}
\author{Berend Schneider\orcidlink{0009-0004-4099-2574}}
\email{berend@uoguelph.ca}
\affiliation{Department of Physics, University of Guelph, Guelph, Ontario, Canada N1G 2W1}

\begin{abstract}
I prove that a collapsing charged scalar field with scalar charge $\mathfrak{e}$ that forms an exactly extremal Reissner--Nordstr\"om black hole à la Kehle--Unger with radius $r_+$ must satisfy $|\mathfrak{e}|r_+ > \frac{1}{3}$. In contrast, I also show that no such bound exists for subextremal collapse: a scalar field with an arbitrary $\mathfrak{e} \ne 0$ can form a subextremal black hole with an arbitrary (subextremal) charge-to-mass ratio. 

Complementing the extremal bound, I prove that a Schwarzschild black hole can become extremal if $|\mathfrak{e}|r_+ > \sqrt{\frac{3 + \sqrt{33}}{3}}$. The fact that $|\mathfrak{e}|r_+$ can be taken to be of order unity could be considered evidence that the third law of black hole mechanics in vacuum is false. 
\end{abstract}

\maketitle


\section{Introduction}

In their seminal paper \cite{KU2022}, Kehle and Unger disproved the third law of black hole mechanics by constructing examples of space-times with arbitrarily high regularity which, starting from space-like initial data, first collapse to a Schwarzschild black hole and then become an exactly extremal Reissner--Nordstr\"om black hole after a finite time. Their construction employs the characteristic gluing method pioneered by Aretakis, Czimek, and Rodnianski in \cite{ACR2021a,ACR2021b,ACR2021c}. The characteristic gluing method works by constructing data on a null hypersurface that is initially Schwarzschild and then goes on to form an exactly extremal Reissner--Nordstr\"om apparent horizon at a finite time. An exact Schwarzschild interior and an exact extremal Reissner--Nordstr\"om exterior may be attached to their respective pieces along this null hypersurface. In the dynamical region interpolating between these two regions, a charged massless scalar field enters the horizon (see also \cite[theorem 2]{KU2024} for a disproof of the third law for charged Vlasov matter). In order for a solution to exist, Kehle and Unger showed that the following condition is \textit{sufficient}: 
\begin{align}
\label{eq:KU_sufficient}
    \Big|\frac{Q_+}{\mathfrak{e}r_+^2}\Big| \ll 1 ,
\end{align}
where $\mathfrak{e}$ is the scalar charge and $Q_+$ and $r_+$ are the charge and radius of the final black hole. In the extremal case, this condition becomes $|\mathfrak{e}|r_+ \gg 1$. $\mathfrak{e}$ is a free parameter, and so can be chosen to be arbitrarily large, so that the left-hand side of Eq.~\eqref{eq:KU_sufficient} can be made arbitrarily small. The goal of this paper is to establish a concrete lower bound for $|\mathfrak{e}|r_+$ in the extremal case $|Q_+| = r_+$, and prove that no such bound exists in the subextremal case. 
\begin{thm}
\label{mainthm:extremal}
In the extremal case $|Q_+| = r_+$, a \underline{necessary} condition for a gluing solution (cf. Definition~\ref{def:solution}) to exist is that
\begin{align}
\label{eq:necessary}
    |\mathfrak{e}|r_+ > \frac{1}{3} .
\end{align}
A \underline{sufficient} condition for a $C^0$ gluing solution to exist is that
\begin{align}
\label{eq:bound_on_sharp_bound}
    |\mathfrak{e}|r_+ > \sqrt{\frac{3 + \sqrt{33}}{3}} \approx 1.70729 .
\end{align}
\end{thm}
\noindent Theorem~\ref{mainthm:extremal} also holds when gluing to a regular centre rather than Schwarzschild (cf. remark~\ref{remark:regular_centre}). The \textit{necessary} condition~\eqref{eq:necessary} complements Kehle and Unger's \textit{asymptotic sufficient} condition~\eqref{eq:KU_sufficient}, and is proven in Theorem~\ref{thm:extremal}. The concrete \textit{sufficient} condition~\eqref{eq:bound_on_sharp_bound} is proven in Theorem~\ref{thm:extremal_limit}, narrowing down the sharp bound to $[\tfrac{1}{3}, \sqrt{\frac{3 + \sqrt{33}}{3}}]$. This is considerably smaller than $|\mathfrak{e}|r_+ = 7.98$, which is the smallest value at which a recent numerical study \cite{GRS2026} found $C^0$ gluing solutions, although that work only considered a limited family of solutions. Eq.~\eqref{eq:bound_on_sharp_bound} will be shown by constructing an explicit $C^0$ gluing solution. In the subextremal case, this solution is valid for any $|Q_+| < r_+$ and $\mathfrak{e} \ne 0$. 
\begin{thm}
\label{mainthm:subextremal}
For any $|Q_+| < r_+$ and $\mathfrak{e} \ne 0$ there exists a $C^0$ gluing solution (cf. Definition~\ref{def:solution}) gluing a Schwarzschild exterior cone with mass $m_0 \ge 0$ to a subextremal Reissner--Nordstr\"om horizon with radius $r_+ > 2m_0$ and charge $Q_+$.
\end{thm}
\noindent This result is proven in Theorem~\ref{thm:subextremal}, and it removes the condition~\eqref{eq:KU_sufficient} entirely for the existence of subextremal $C^0$ gluing solutions. 

It has been shown in \cite{Reall2025,MR2025} that extremal Reissner--Nordstr\"om black holes cannot form in gravitational collapse if the matter satisfies a certain local mass-charge inequality. The matter considered in this work---the Einstein--Maxwell--charged scalar field system (cf. Eq.~\eqref{eq:EMcS} below)---does not satisfy this mass-charge inequality, and this work therefore complements these previous results. Indeed, it is precisely the fact that an arbitrarily high amount of charge can be generated per unit mass that allows $|\frac{Q_+}{\mathfrak{e}r_+^2}|$ to be arbitrarily large in the subextremal case. 

\subsection{Relevance for the third law of black hole mechanics in vacuum}

The problem of dynamically forming charged extremal black holes with small $|\mathfrak{e}|r_+$ is important in view of the third law in vacuum. For the charged system, $|\mathfrak{e}|^{-1}$ sets the scale, but in vacuum there is no analogous parameter to compare with the black hole mass. Therefore, \textit{any argument that only works with $\mathfrak{e}r_+ \gg 1$ has no analogue in vacuum}. Proving that $|\mathfrak{e}|r_+$ may be taken to be of order unity, which is done in Theorem~\ref{thm:extremal_limit} below, means that there is no structural reason that $|\mathfrak{e}|r_+$ \textit{had} to be chosen large. The sufficient condition~\eqref{eq:bound_on_sharp_bound} could therefore be considered evidence that the third law in vacuum is false. 

A problem that is strikingly similar to charged scalar collapse to an extremal Reissner--Nordstr\"om black hole in four space-time dimensions is the problem of \textit{vacuum} collapse to an extremal \textit{Myers--Perry} black hole in \textit{five} space-time dimensions. Because Myers--Perry is much more symmetric than Kerr, the Einstein equations on the horizon describing extremal collapse can be taken to be ODEs, and are very similar in structure to the Einstein--Maxwell--charged scalar field system in four dimensions. This system was constructed and integrated numerically in \cite{CGRS2026}, providing strong evidence that the third law is false in five-dimensional vacuum general relativity. The obstruction to finding an analytic proof using gluing techniques à la Kehle--Unger is precisely the lack of a `small' parameter like $|\mathfrak{e}|^{-1}$.

\subsection{Outline of the paper}

The structure of the paper is as follows: the Einstein--Maxwell--charged scalar field equations on the horizon will be introduced in \S\ref{sec:equations}, together with all gauge choices and a definition of gluing solutions. The next two sections prove the main theorems: \S\ref{sec:subextremal} deals with the subextremal case, where it is proven in Theorem~\ref{thm:subextremal} that $C^0$ gluing solutions exist for all (subextremal) charge-to-mass ratios and $\mathfrak{e} \ne 0$. Theorem~\ref{mainthm:extremal} is proven in two parts. The extremal limit of Theorem~\ref{thm:subextremal}'s construction is explored in \S\ref{sec:failure_at_extremality}, which proves the sufficient condition~\eqref{eq:bound_on_sharp_bound} in Theorem~\ref{thm:extremal_limit}. The lower bound $|\mathfrak{e}|r_+ > \tfrac{1}{3}$ is proven in Theorem~\ref{thm:extremal} in \S\ref{sec:extremal}.

\section{The system of equations}
\label{sec:equations}

I will be using the same gauge choices and follow the same conventions as Kehle--Unger~\cite{KU2022}. For completeness I will restate all the relevant details here. 

The Einstein--Maxwell--charged scalar field system is given in covariant form by
\begin{align}
\label{eq:EMcS}
    R_{ab} &= 2\operatorname{Re}(D_a\phi\overline{D_b\phi}) + 2(F_{ac}F_b^{\;\;c} - \tfrac{1}{4}g_{ab}F^{cd}F_{cd}) , \nonumber \\
    \nabla^aF_{ab} &= 2\mathfrak{e}\operatorname{Im}(\phi\overline{D_b\phi}) , \tag{EMcS} \\
    D^aD_a\phi &= 0 , \nonumber
\end{align}
where $\phi$ is a complex scalar field, $F_{ab} = \nabla_aA_b - \nabla_bA_a$ is the electromagnetic field given in terms of the potential $A_a$, and $D_a = \nabla_a + i\mathfrak{e}A_a$ is the gauge-covariant derivative with coupling constant $\mathfrak{e} \ne 0$. The solutions I will be considering are spherically symmetric and contain an event horizon $\mathcal{H}^+$. The metric is given by
\begin{align}
    ds^2 = -\Omega^2dudv + r^2dS^2 ,
\end{align}
where $\Omega|_{\mathcal{H}^+} = 1$, $dS^2$ is the unit two-sphere metric, and $u,v$ are ingoing and outgoing null coordinates---$u$ being constant on $\mathcal{H}^+$. I am partially fixing the electromagnetic gauge by setting $A_v = 0$. Throughout this paper I will be concerned with dynamical quantities on the horizon, i.e. quantities defined on $\mathcal{H}^+$ that are functions of $v$ alone. From this point onward, $'$ will denote $\partial_v$. The horizon is an exact Schwarzschild exterior cone at early times, which I will take to be $v \le 0$ using the gauge freedom $v \to v + \Delta$, with Hawking mass $m_0 = \frac{r}{2}(1 + 4r'\partial_ur)$, $r > 2m_0 \ge 0$, $r' > 0$. At late times, which I will take to be $v \ge 1$ using the gauge freedom $(u, v) \to (\beta^{-2}u, \beta^2v)$, the horizon will be an exact Reissner--Nordstr\"om apparent horizon, with constant radius $r_+$, and negative ingoing expansion $\partial_ur < 0$. The dynamics, then, are confined to the interval $[0,1]$. The radius $r$ satisfies Raychaudhuri's equation
\begin{align}
\label{eq:Raychaudhuri}
    r'' = -r|\phi'|^2 , \tag{R}
\end{align}
with boundary conditions 
\begin{align}
\label{eq:boundary_conditions}
    r(0) = r_0 > 2m_0 \ge 0 , && r(1) = r_+ > r_0 , && r'(1) = 0 . \tag{BC}
\end{align}
The scalar field satisfies $\phi(0) = \phi(1) = 0$ (chosen to match $\phi = 0$ on the Schwarzschild and Reissner--Nordstr\"om regions $v \le 0$ and $v \ge 1$). The charge $Q$ satisfies Maxwell's equation
\begin{align}
\label{eq:Maxwell}
    Q(v) = \mathfrak{e}\int_0^vr^2\operatorname{Im}(\phi\overline{\phi'})\dd x , \tag{M}
\end{align}
with $Q_+ = Q(1)$. The ingoing null expansion $\partial_ur$ satisfies the Einstein equation (cf.~\cite[Eq. (2.12)]{KU2022})
\begin{align}
    (-4r\partial_ur)' = 1 - \frac{Q^2}{r^2} , 
\end{align}
with boundary conditions $(-4r\partial_ur)(0) = f_0\frac{r_0}{r'(0)}$ where $f_0 = 1 - \frac{2m_0}{r_0}$ satisfies $0 < f_0 \le 1$. The horizon must not be anti-trapped, so that $\partial_ur < 0$ on the entire interval. This is a constraint: 
\begin{align}
    0 < -4r\partial_ur = v - \int_0^v\frac{Q^2}{r^2}\dd x + (-4r\partial_ur)(0) = v - \int_0^v\frac{Q^2}{r^2}\dd x + f_0\frac{r_0}{r'(0)} .
\end{align}
It follows that
\begin{align}
\label{eq:constraint}
    \int_0^v\frac{Q^2}{r^2}\dd x < v + f_0\frac{r_0}{r'(0)} . \tag{C}
\end{align}
I will refer to Eq.~\eqref{eq:constraint} as \textit{the constraint}. 

\begin{definition}
\label{def:solution}
The functions $(r, \phi) \in C^1([0,1]) \times (C^0([0,1]) \cap H^1([0,1]))$ are said to be a $C^0$ gluing solution (often simply referred to as `solution') if they satisfy Raychaudhuri's equation~\eqref{eq:Raychaudhuri} with boundary conditions~\eqref{eq:boundary_conditions} and $\phi(0) = \phi(1) = 0$, and the constraint~\eqref{eq:constraint} for all $v \in [0,1]$, where $Q$ is given by Eq.~\eqref{eq:Maxwell}. The final charge is given by $Q_+ = Q(1)$. 
\end{definition}

\begin{remark}
The regularity assumption $\phi \in H^1([0,1])$ means that $\phi' \in L^2([0,1])$, which implies that $\int_0^1|\phi'|^2\dd v < \infty$. This is necessary in order for $(r, \phi)$ to be a weak solution to Raychaudhuri's equation~\eqref{eq:Raychaudhuri}. 
\end{remark}

Solutions, as defined by Definition~\ref{def:solution}, are used by Kehle and Unger in \cite{KU2022} to construct $C^0$ space-times that dynamically form exact Reissner--Nordstr\"om black holes.

\section{Constructing gluing solutions}
\label{sec:subextremal}

In this section I will prove that subextremal gluing solutions exist for any charge-to-mass ratio and any $\mathfrak{e} \ne 0$. When $|\mathfrak{e}|r_+$ is small, a major problem will be to generate enough charge. This problem is dealt with in \S\ref{sec:self_similar}, where I construct a self-similar solution to Eqs.~\eqref{eq:Raychaudhuri} and~\eqref{eq:Maxwell} that can generate any amount of charge. In \S\ref{sec:three_piece_solution} the main result is proven by constructing an explicit three-piece solution that uses the charge-generating solution of \S\ref{sec:self_similar} sandwiched between two solutions that satisfy the endpoint conditions $\phi(0) = \phi(1) = 0$ and boundary conditions~\eqref{eq:boundary_conditions}. In \S\ref{sec:failure_at_extremality} I will take the extremal limit of the solution constructed in \S\ref{sec:three_piece_solution}, and show that the constraint~\eqref{eq:constraint} can only be satisfied for $|\mathfrak{e}|r_+ > \sqrt{\frac{3 + \sqrt{33}}{3}}$. This will prove the sufficient statement~\eqref{eq:bound_on_sharp_bound} of Theorem~\ref{mainthm:extremal}. 

\subsection{A self-similar solution}
\label{sec:self_similar}

Since $\operatorname{Im}(\phi\overline{\phi'}) \ne 0$ implies that $|\phi'| \ne 0$, whenever charge is generated $r$ is simultaneously focused. This is a major obstruction, because for large $\frac{Q_+}{\mathfrak{e}r_+^2}$ this means that care must be taken to generate enough charge per unit focusing. This subsection will deal with this problem. 

Concretely, decomposing $\phi = |\phi|e^{i\theta}$, Raychaudhuri's equation~\eqref{eq:Raychaudhuri} and Maxwell's equation~\eqref{eq:Maxwell} become
\begin{subequations}
\begin{align}
    r'' &= -r\big[(|\phi|')^2 + |\phi|^2(\theta')^2\big] , \\
    Q' &= -\mathfrak{e}r^2|\phi|^2\theta' .
\end{align}
\end{subequations}
The growth in $Q$ is linear in $\theta'$, while the focusing of $r$ is quadratic in $\theta'$. In order to generate a lot of charge, $|\theta'|$ and $|\mathfrak{e}|$ cannot simultaneously be small, and it therefore becomes necessary to strongly focus $r$. The solution to this problem is to take $r'$ large on a small interval so that $r$ can remain large for long. 

Inspired by this observation, consider the following \textit{self-similar solution}: Fix $\epsilon > 0$ and let, for $v > \epsilon$, 
\begin{subequations}
\label{eq:self-similar}
\begin{align}
    r &= Rv^s , \\
    \phi &= \Phi e^{-i\lambda\log v} ,
\end{align}
\end{subequations}
where $R$, $s$, $\Phi$, and $\lambda$ are constant, and where $s(1 - s) = \Phi^2\lambda^2$ so that $r$ solves Eq.~\eqref{eq:Raychaudhuri}. When $s \ll 1$, $r \approx R$ while at the same time both $r'$ and $\theta'$ rapidly grow as $v \to \epsilon$ if $\epsilon$ was chosen small. Maxwell's equation~\eqref{eq:Maxwell} becomes
\begin{align}
    Q' = \mathfrak{e}\Phi^2\lambda\frac{r^2}{v} ,
\end{align}
which, with initial condition $Q(\epsilon) = 0$, has the solution
\begin{align}
    Q = \frac{\mathfrak{e}\Phi^2\lambda}{2s}(r^2 - r^2(\epsilon)) \approx \mathfrak{e}\Phi^2\lambda R^2\log(v/\epsilon) ,
\end{align}
for small $s$. Hence, $Q$ can be grown arbitrarily large by letting $\epsilon$ become arbitrarily small. 

The self-similar solution~\eqref{eq:self-similar} solves the charge-generation problem, but it remains to be shown that the boundary conditions~\eqref{eq:boundary_conditions} and the constraint~\eqref{eq:constraint} can be satisfied. In the next part, I will show that both conditions can be satisfied in the subextremal case.

\subsection{The subextremal case \texorpdfstring{$|Q_+| < r_+$}{|Q| < r}}
\label{sec:three_piece_solution}

\begin{theorem}[Theorem~\ref{mainthm:subextremal}]
\label{thm:subextremal}
For any $|Q_+| < r_+$, $0 \le 2m_0 < r_0 < r_+$, and $\mathfrak{e} \ne 0$ there exists a solution $(r, \phi)$ (cf. Definition~\ref{def:solution}).
\end{theorem}

The idea of the proof is as follows: split the interval into three parts: $I_\text{on} \cup I_\text{charging} \cup I_\text{off}$. On the middle part $I_\text{charging}$, choose the self-similar solution~\eqref{eq:self-similar} to generate as much charge as needed. On $I_\text{on}$ and $I_\text{off}$, the field will be turned on and off so that $\phi(0) = \phi(1) = 0$, in a way that also makes $r$ satisfy the boundary conditions~\eqref{eq:boundary_conditions} and regularity conditions imposed by Definition~\ref{def:solution}. See Fig.~\ref{fig:thm_1_solution} for an example solution. This is possible because if $|\phi|$---which is a constant, free parameter on the charge-generating part---is chosen small, then $I_\text{on}$ and $I_\text{off}$ will contribute negligibly to $r$. It will remain to be shown that the constraint~\eqref{eq:constraint} can be satisfied. I will do this by choosing $r_0 > Q_+$, so that $Q/r < 1$ everywhere. 

\begin{figure}[tb]
  \centering
  \begin{subfigure}[b]{0.48\textwidth}
    \centering
    \includegraphics[width=\linewidth]{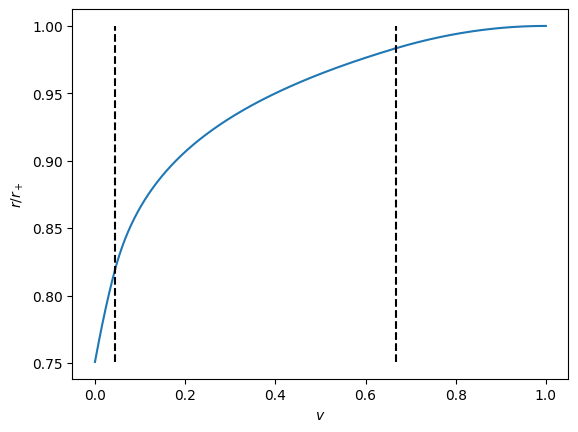}
  \end{subfigure}
  \hfill
  \begin{subfigure}[b]{0.48\textwidth}
    \centering
    \includegraphics[width=\linewidth]{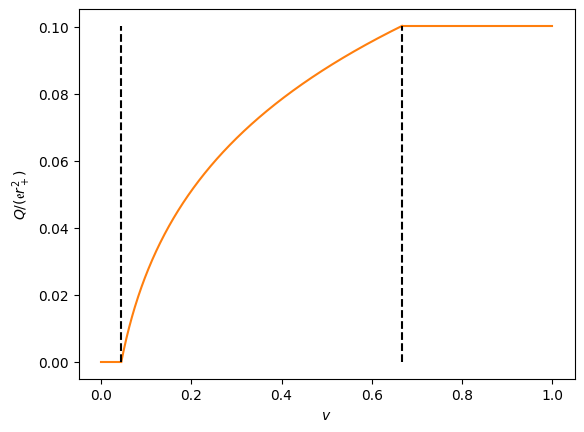}
  \end{subfigure}
  \caption{The radius $r/r_+$ (left) and charge $Q/(\mathfrak{e}r_+^2)$ (right) of an example of a solution constructed in the proof of Theorem~\ref{thm:subextremal} with $r_0/r_+ = 0.75$ and $Q_+/(\mathfrak{e}r_+^2) = 0.1$. The dashed lines indicate the boundary between $I_\text{on}$ and $I_\text{charging}$, and $I_\text{charging}$ and $I_\text{off}$. Charge is generated on $I_\text{charging}$ using the self-similar solution~\eqref{eq:self-similar}. On $I_\text{on}$ and $I_\text{off}$, $\phi$ goes linearly to zero to satisfy the endpoint conditions.}
  \label{fig:thm_1_solution}
\end{figure}

\begin{proof}
Start by assuming that $r_0$ satisfies $Q_+ < r_0 < r_+$. I will return to the case $r_0 \le Q_+$ at the end of the proof. Given any solution, replacing $\phi \mapsto \bar{\phi}$ or $\mathfrak{e} \mapsto -\mathfrak{e}$ changes the sign of the charge: $Q \mapsto -Q$, so without loss of generality assume $Q_+ > 0$ and $\mathfrak{e} > 0$. Let $0 < v_1 < v_2 < 1$. Define the piecewise smooth solution
\begin{align}
\label{eq:solution_r}
    r = 
    \begin{cases}
    r_+\cos(\Phi)\sec(\alpha)\big(\frac{v_1}{v_2}\big)^s\cos\big(\Phi\frac{v_1 - v}{v_1} + \alpha\big) & v \in [0, v_1] = I_\text{on} \\
    r_+\cos(\Phi)\big(\frac{v}{v_2}\big)^s & v \in [v_1, v_2] = I_\text{charging} \\
    r_+\cos\big(\Phi\frac{1-v}{1-v_2}\big) & v \in [v_2, 1] = I_\text{off}
    \end{cases}
\end{align}
and
\begin{align}
\label{eq:solution_phi}
    \phi = 
    \begin{cases}
    \Phi e^{-i\frac{c}{\Phi}\log v_1}\frac{v}{v_1} & v \in [0, v_1] = I_\text{on} \\
    \Phi e^{-i\frac{c}{\Phi}\log v} & v \in [v_1, v_2] = I_\text{charging} \\
    \Phi e^{-i\frac{c}{\Phi}\log v_2}\frac{1-v}{1-v_2} & v \in [v_2, 1] = I_\text{off}
    \end{cases}
\end{align}
where $(c, \alpha, v_2)$ are given by
\begin{subequations}
\begin{align}
    c &= \sqrt{s(1 - s)} , \\
    \alpha &= \arctan\Big(\frac{s}{\Phi}\Big) , \\
    v_2 &= \frac{1}{1 + \frac{\Phi}{s}\tan(\Phi)} .
\end{align}
\end{subequations}
$c$ is chosen such that $(r,\phi)$ satisfy Raychaudhuri's equation~\eqref{eq:Raychaudhuri}, $\alpha$ is chosen such that $r'$ is continuous at $v_1$, and $v_2$ is chosen such that $r'$ is continuous at $v_2$. The remaining quantities $(s, \Phi, v_1)$ are free parameters with $0 < s < 1$, $0 < \Phi < \frac{\pi}{2} - \alpha$, and $0 < v_1 < v_2$. 

Charge is only generated on $I_\text{charging}$, and can be calculated to be
\begin{align}
\label{eq:scalar_charge}
    \frac{Q_+}{\mathfrak{e}r_+^2} = c\Big(\frac{1 - \big(\frac{v_1}{v_2}\big)^{2s}}{2s}\Big)\Phi\cos^2(\Phi) .
\end{align}
The ratio $r_+/r_0$ is given by
\begin{align}
    \log\Big(\frac{r_+}{r_0}\Big) = \underbrace{\log\big(\cos(\alpha)\sec(\Phi + \alpha)\big)}_\text{on} + \underbrace{s\log\Big(\frac{v_2}{v_1}\Big)}_\text{charging} + \underbrace{\log\sec(\Phi)}_\text{off} .
\end{align}
I will first check that $\frac{Q_+}{\mathfrak{e}r_+^2}$ can be chosen arbitrarily and $\log(r_+/r_0)$ can be made arbitrarily small simultaneously, while also satisfying the constraint. 

Define $v_1 = v_2e^{-\Phi/s}$. This leaves $\Phi$ and $s$ as free parameters. For small $\Phi$, 
\begin{align}
    \log\Big(\frac{r_+}{r_0}\Big) = s\log\Big(\frac{v_2}{v_1}\Big) + O(\Phi^2 + s) = \Phi + O(\Phi^2 + s) ,
\end{align}
so that for small $\Phi$ and small $s$, $\log(r_+/r_0)$ can be made arbitrarily small. With $\Phi$ fixed, $s$ defines the continuous map
\begin{align}
\label{eq:s_to_Q_map}
    s \mapsto \frac{Q_+}{\mathfrak{e}r_+^2} = \frac{1}{2}\sqrt{\frac{1-s}{s}}(1 - e^{-2\Phi})\Phi\cos^2(\Phi) = \sqrt{\frac{1-s}{s}}\Phi^2(1 + O(\Phi)) ,
\end{align}
which ranges from $\frac{Q_+}{\mathfrak{e}r_+^2} \to 0$ as $s \to 1$ to $\frac{Q_+}{\mathfrak{e}r_+^2} \to \infty$ as $s \to 0$. By the intermediate value theorem, for any $0 < \frac{Q_+}{\mathfrak{e}r_+^2} < \infty$ there exists a solution $s$. By choosing $\Phi$ sufficiently small, the solution $s$ for any given $\frac{Q_+}{\mathfrak{e}r_+^2}$ will also be small, so that $\log(r_+/r_0)$ can be chosen small simultaneously. 

Since $Q \le Q_+$ and $r \ge r_0$, $\frac{Q}{r} \le \frac{Q_+}{r_0}$. Now, because $r_0$ was picked such that $Q_+ < r_0 < r_+$, it follows that $\frac{Q}{r} < 1$ everywhere, so that $\int_0^v\frac{Q^2}{r^2}\dd x < v < v + f_0\frac{r_0}{r'(0)}$. The constraint is therefore satisfied. 

Finally, consider the case where $r_0$ is smaller. Such a solution can be obtained in a straightforward manner from any solution $\tilde r$ satisfying $\tilde Q_+ < \tilde r_0 < \tilde r_+$: define a solution $r$ with $2m_0 < r_0 < \tilde r_0$ as
\begin{align}
    r = 
    \begin{cases}
    r_0 + v(1 - \gamma^{-1})^{-1}(\tilde r_0 - r_0) & 0 \le v \le 1 - \gamma^{-1} \\
    \tilde r(\gamma(v - 1) + 1) & 1 - \gamma^{-1} \le v \le 1
    \end{cases}
\end{align}
where $\gamma = 1 + (\tilde r_0 - r_0)/\tilde r'(0)$ is chosen such that $r'$ is continuous across $v = 1 - \gamma^{-1}$. $\phi$ is defined similarly, chosen to vanish identically on $0 \le v \le 1 - \gamma^{-1}$. It is easy to verify that $r, \phi$ solve Eqs.~\eqref{eq:Raychaudhuri} and~\eqref{eq:Maxwell} with the same final charge $Q_+ = \tilde Q_+$, and that $r$ satisfies the boundary conditions~\eqref{eq:boundary_conditions} with $r(1) = \tilde r(1) > \tilde r_0 > r_0$. Since no charge is generated on the appended interval, $\frac{Q}{r} \le \frac{Q}{\tilde r_0} < 1$ so that the constraint~\eqref{eq:constraint} is satisfied. 
\end{proof}

\begin{remark}
\label{remark:C1_modification}
With small modifications, the solutions constructed in Theorem~\ref{thm:subextremal} can be turned into $C^1$ gluing solutions (cf. \cite[definition 3.1]{KU2022}). These solutions require stronger regularity and an additional constraint. Firstly, $\phi$ needs to be $C^1$ instead of $C^0 \cap H^1$, which can be achieved by softening the edges of $\phi$ at $0$, $v_1$, $v_2$ and $1$. Secondly, the constraint $\int_0^1\partial_v(\partial_u\phi)\dd v = 0$ needs to hold, which can be shown to be equivalent to
\begin{align}
\label{eq:C1_constraint}
    \int_0^1\Big(\frac{2m}{r^2} - \frac{Q^2}{r^3}\Big)\phi\dd v = 0 ,
\end{align}
where $m$ is the Hawking mass. For Theorem~\ref{thm:subextremal}'s solutions the left-hand side of Eq.~\eqref{eq:C1_constraint} is of the order $\Phi$, which can be cancelled by slightly modifying the solution, adding a small bump with constant phase after $I_\text{off}$ (where $I_\text{off}$ is suitably shortened and modified to make space for this bump). This bump does not contribute to $Q$ and only contributes $O(\Phi^2)$ to $r$, so $\log(r_+/r_0)$ can be kept small. 
\end{remark}
The solutions constructed in Theorem~\ref{thm:subextremal} used subextremality in an essential way: by being able to choose $Q_+ < r_0 < r_+$, the constraint was trivially satisfied. For large $\frac{Q_+}{\mathfrak{e}r_+^2}$ the constraint is quite strong. Indeed, $\frac{r_0}{r'(0)}$ is given by
\begin{align}
    \frac{r_0}{r'(0)} = \frac{v_1}{\Phi}\cot(\Phi + \alpha) ,
\end{align}
which goes to zero \textit{exponentially} as $\frac{Q_+}{\mathfrak{e}r_+^2}$ gets large, since $v_1 \approx \exp(-\frac{1}{c\Phi}\frac{Q_+}{\mathfrak{e}r_+^2})$. At the same time, at extremality it becomes impossible to have $\frac{Q}{r} < 1$ everywhere, so the fact that the constraint becomes strong matters.

\subsection{The extremal limit of the construction}
\label{sec:failure_at_extremality}

Here I will examine the extremal case $Q_+ = r_+$ of the solutions constructed in Theorem~\ref{thm:subextremal}. I will show that this family of solutions has a minimum value of $\mathfrak{e}r_+$, below which no solutions exist. The existence of a lower bound will turn out to not be a feature of this family of solutions alone, but a general property of all solutions. A concrete lower bound on $|\mathfrak{e}|r_+$, valid for \textit{all} solutions, will be proven in the next section. The example here will set an upper bound for the sharp constant bounding $|\mathfrak{e}|r_+$, and serves as inspiration for the next section. 

I will frame this result as a more general theorem on the existence of $C^0$ gluing solutions. 
\begin{theorem}[Theorem~\ref{mainthm:extremal}, sufficient part]
\label{thm:extremal_limit}
In the extremal case $|Q_+| = r_+$ a sufficient condition for $C^0$ gluing solutions to exist (cf. Definition~\ref{def:solution}) is that
\begin{align}
    |\mathfrak{e}|r_+ > \sqrt{\frac{3 + \sqrt{33}}{3}} \approx 1.70729 .
\end{align}
\end{theorem}

\begin{proof}
Again, without loss of generality, let $Q_+ > 0$ and $\mathfrak{e} > 0$. Take the solution~\eqref{eq:solution_r} and~\eqref{eq:solution_phi} from the proof of Theorem~\ref{thm:subextremal}, and fix $Q_+ = r_+$. The parameters $(c, \alpha, v_2)$ are again given by $c = \sqrt{s(1 - s)}$, $\alpha = \arctan(\frac{s}{\Phi})$, and $v_2 = (1 + \frac{\Phi}{s}\tan(\Phi))^{-1}$. Fix $v_1 = v_2e^{\log(\Phi)/s}$ (note that this choice is different from the one in the proof of Theorem~\ref{thm:subextremal}). The free parameters are $(s, \Phi)$, where $0 < s < 1$ and $0 < \Phi < \min(1, \frac{\pi}{2} - \alpha)$ ($\Phi < 1$ is chosen to ensure that $v_1 < v_2$). Then, similar to Eq.~\eqref{eq:s_to_Q_map}, the solution defines the map
\begin{align}
\label{eq:s_map_extremal}
    s \mapsto \mathfrak{e}r_+ = 2\sqrt{\frac{s}{1 - s}}\frac{\sec^2(\Phi)}{\Phi}\Big(1 - \Big(\frac{v_1}{v_2}\Big)^{2s}\Big)^{-1} = 2\sqrt{\frac{s}{1 - s}}\frac{1}{\Phi}(1 + O(\Phi^2)) ,
\end{align}
which is increasing. The claimed lower bound for this solution will follow from a lower bound on $s$, which will come from the constraint~\eqref{eq:constraint}. The charge on $[v_1, v_2] = I_\text{charging}$ is given by
\begin{align}
\label{eq:charge_estimate}
    Q = r_+\frac{v^{2s} - v_1^{2s}}{v_2^{2s} - v_1^{2s}} < r_+\Big(\frac{v}{v_2}\Big)^{2s} .
\end{align}
Claim: it only needs to be checked that the constraint~\eqref{eq:constraint} is satisfied at $v = 1$. To prove this claim, first define
\begin{align}
    G(v) = \int_0^v\frac{Q^2}{r^2}\dd x - v ,
\end{align}
so that the constraint becomes $G(v) < f_0\frac{r_0}{r'(0)}$. The function $G$ is initially zero and decreasing, and then starts to increase up to $v = 1$. Indeed, on $I_\text{on}$, $Q = 0$ so that $G' = -1 < 0$. Then, on $I_\text{charging}$, $Q/r \propto v^s - v_1^{2s}v^{-s}$ so that $G'$ is increasing. Finally, on $I_\text{off}$, $Q = r_+$ and $r \le r_+$ and therefore $G' \ge 0$. Hence, $G$ must have a maximum at $v = 0$ or $v = 1$. This means that the constraint is satisfied everywhere if the constraint is satisfied at both $v = 0$ and $v = 1$, and since the constraint is satisfied at $v = 0$, this proves the claim. 

From Eq.~\eqref{eq:charge_estimate} it follows that
\begin{multline}
    \int_0^1\frac{Q^2}{r^2}\dd v = \int_{v_1}^{v_2}\frac{Q^2}{r^2}\dd v + \int_{v_2}^1\frac{Q^2}{r^2}\dd v \le  \int_0^{v_2}\frac{Q^2}{r^2}\dd v + \int_{v_2}^1\frac{Q^2}{r^2}\dd v \\
    \le \sec^2(\Phi)\int_0^{v_2}\Big(\frac{v}{v_2}\Big)^{2s}\dd v + \int_{v_2}^1\sec^2\Big(\Phi\frac{1-v}{1-v_2}\Big)\dd v = v_2\Big(\frac{\sec^2(\Phi)}{2s + 1} + \frac{\tan^2(\Phi)}{s}\Big) .
\end{multline}
Recall that $v_2 = s(s + \Phi\tan(\Phi))^{-1}$. The constraint is satisfied if
\begin{align}
\label{eq:extremal_constraint}
    \frac{s\sec^2(\Phi)}{2s+1} + \tan^2(\Phi) < s + \Phi\tan(\Phi) .
\end{align}
Since $\tan(\Phi) > \Phi$ for small $\Phi$, Eq.~\eqref{eq:extremal_constraint} is violated for small $s$, but for $s$ larger than some critical value $s_*$ Eq.~\eqref{eq:extremal_constraint} holds. To find the limit, let $s_* = t\Phi^2$ and expand the constraint in $\Phi$: 
\begin{align}
    (1 + t)\Phi^2 + (\tfrac{2}{3} + (1 - 2t)t)\Phi^4 + O(\Phi^6) < (1 + t)\Phi^2 + \tfrac{1}{3}\Phi^4 + O(\Phi^6) .
\end{align}
The $\Phi^2$ terms cancel, and the $\Phi^4$ coefficient yields $s_* = t\Phi^2 = \frac{3 + \sqrt{33}}{12}\Phi^2 + O(\Phi^4)$. The scalar charge becomes (cf. Eq.~\eqref{eq:s_map_extremal})
\begin{multline}
    \mathfrak{e}r_+ > 2\sqrt{\frac{s_*}{1 - s_*}}\frac{\sec^2(\Phi)}{\Phi}\Big(1 - \Big(\frac{v_1}{v_2}\Big)^{2s}\Big)^{-1} = 2\sqrt{t}(1 + O(\Phi^2)) \\
    = \sqrt{\frac{3 + \sqrt{33}}{3}}(1 + O(\Phi^2)) .
\end{multline}
By choosing $\Phi$ sufficiently small, $\mathfrak{e}r_+$ can get arbitrarily close to $\sqrt{\frac{3 + \sqrt{33}}{3}}$. By choosing $s \to 1$, $\mathfrak{e}r_+ \to \infty$. Hence, by the intermediate value theorem extremal solutions exist with $\mathfrak{e}r_+ \in (\sqrt{\frac{3 + \sqrt{33}}{3}}, \infty)$. 
\end{proof}

For this reason, the possibility that solutions may not exist for small $\mathfrak{e}r_+$ in the extremal case $Q_+ = r_+$ is not ruled out by Theorem~\ref{thm:subextremal}. Indeed, in the next section I will prove that no analogue of Theorem~\ref{thm:subextremal} holds at extremality. 
\begin{remark}
In the subextremal case, the left-hand side of Eq.~\eqref{eq:extremal_constraint} is multiplied by $(Q_+/r_+)^2 < 1$. For small $s$, the constraint becomes $(Q_+/r_+)^2\tan(\Phi) < \Phi$ which can be satisfied by a sufficiently small $\Phi$. 
\end{remark}

\section{The extremal case \texorpdfstring{$|Q_+| = r_+$}{|Q| = r}}
\label{sec:extremal}

In this section I will prove that in the extremal case, $|\mathfrak{e}|r_+ > \frac{1}{3}$. The intuition behind the proof comes from understanding why the construction of Theorem~\ref{thm:subextremal} fails at extremality. The charge is given by the following expression (cf. Proposition~\ref{prop:exact_identity} below):
\begin{align}
    Q_+ = \underbrace{\mathfrak{e}\int_0^1vr^2\operatorname{Im}(\phi\overline{\phi'})\dd v}_{\text{weighted charge production}} + \underbrace{\int_0^1Q\dd v}_{\text{bulk charge}} .
\end{align}
When $|\mathfrak{e}|r_+$ is small, $r'(0)$ will be large (this is seen in the construction in Theorem~\ref{thm:subextremal} and will be proven in Proposition~\ref{prop:weighted_charge_production_bound} below), and this will limit the size of the bulk charge term through the constraint. In order for the bulk term to be large, the charge has to be large on most of the interval, but because of the constraint this is only possible if $r$ is similarly large, and therefore has to grow early. It will be shown that $r$ growing early bounds the weighted charge production term, and that if $|\mathfrak{e}|r_+$ is too small, this bound is too tight to close the gap left by the bulk charge term. 

\subsection{Basic estimates and preliminary lemmas}

\begin{definition}
Define the following moments: 
\begin{align}
    E &:= \frac{1}{r_0}\int_0^1r|\phi'|^2\dd v , \\
    H &:= \frac{1}{r_+}\int_0^1v^2r|\phi'|^2\dd v , \\
    P &:= \frac{1}{r_+}\int_0^1v^2r(|\phi|')^2\dd v . 
\end{align}
\end{definition}
The interpretation of these moments comes from the next proposition. In physical terms, $E$ measures the strength of the constraint, while $H$ measures how large $r$ is on the whole interval; if $r \approx r_+$ then $H \to 0$, while if $r \approx r_0 + v(r_+ - r_0)$, $H \to 1 - \frac{r_0}{r_+}$. $P$ measures how much of $H$ does not generate charge. This is because $\operatorname{Im}(\phi\overline{\phi'}) = -|\phi|^2\theta'$, so the derivative of the magnitude $|\phi|$ does not contribute to the charge current. 
\begin{proposition}[Elementary consequences of Raychaudhuri's equation~\eqref{eq:Raychaudhuri}]\;
\label{prop:elementary}
\begin{enumerate}[label=(\roman*)]
    \item Because $r'' \le 0$, $r$ has the basic concavity bounds $0 < r_0 \le r \le r_+$. 
    \item Concavity implies $r'(0) > r_+ - r_0 > 0$ and $r' \ge 0$. 
    \item $E = \frac{r'(0)}{r_0}$, since $r_0E = -\int_0^1r''\dd v = -r'\big|_0^1 = r'(0)$. 
    \item $H = 2 - \frac{2}{r_+}\int_0^1r\dd v$, which follows by integrating by parts twice: 
    \begin{multline}
        r_+H = \int_0^1v^2(-r'')\dd v = -(v^2r')\big|_0^1 + 2\int_0^1vr'\dd v \\
        = 2(vr)\big|_0^1 - 2\int_0^1r\dd v = 2r_+ - 2\int_0^1r\dd v . 
    \end{multline}
    From concavity, $r_0 + v(r_+ - r_0) \le r \le \min(r_+, r_0 + vr_0E)$, so that $0 < H \le 1 - \frac{r_0}{r_+} < 1$. 
    \item $P \le H$, since $(|\phi|')^2 \le |\phi'|^2$. 
\end{enumerate}
\end{proposition}

\begin{lemma}
\label{lemma:convex}
Let $f \ge 0$ be absolutely continuous and convex on $(0,1)$, and let $f(1) = 0$. Then
\begin{align}
    \int_0^1f^2\dd v \le \frac{2}{3}f(0)\int_0^1f\dd v .
\end{align}
\end{lemma}

\begin{proof}
    By convexity, $f(v) \le f(0) + vf'(v)$, so that
    \begin{multline}
        \int_0^1f^2\dd v \le \int_0^1\big[f(0) + vf'(v)\big]f(v)\dd v = f(0)\int_0^1f\dd v + \frac{1}{2}\int_0^1v(f^2)'\dd v \\
        = f(0)\int_0^1f\dd v + (\tfrac{1}{2}vf^2)\big|_0^1 - \frac{1}{2}\int_0^1f^2\dd v = f(0)\int_0^1f\dd v - \frac{1}{2}\int_0^1f^2\dd v ,
    \end{multline}
    using $f(1) = 0$ to discard the boundary term. The result follows by rearranging. 
\end{proof}

\begin{corollary}
\label{cor:concave}
Let $r$ be a solution to Raychaudhuri's equation~\eqref{eq:Raychaudhuri} with boundary conditions~\eqref{eq:boundary_conditions}. Then
\begin{align}
    \int_0^1r^2\dd v \le r_+^2\Big(1 - \frac{1}{3}\Big(2 + \frac{r_0}{r_+}\Big)H\Big) .
\end{align}
\end{corollary}

\begin{proof}
Apply Lemma~\ref{lemma:convex} to $f = r_+ - r$, so that
\begin{align}
    \int_0^1(r_+ - r)^2\dd v \le \frac{2}{3}(r_+ - r_0)\int_0^1(r_+ - r)\dd v = \frac{r_+^2}{3}\Big(1 - \frac{r_0}{r_+}\Big)H ,
\end{align}
where I used that $\int_0^1r\dd v = r_+(1 - \frac{H}{2})$. It follows that
\begin{align}
    \int_0^1r^2\dd v = \int_0^1(r_+ - r)^2\dd v + r_+^2(1 - H) \le r_+^2\Big(1 - \frac{1}{3}\Big(2 + \frac{r_0}{r_+}\Big)H\Big) .
\end{align}
\end{proof}

The following Hardy-like inequality will be used later: 
\begin{lemma}[Hardy]
\label{lemma:Hardy}
Any solution $(r, \phi)$ (cf. Definition~\ref{def:solution}) satisfies
\begin{align}
    \int_0^1r|\phi|^2\dd v \le 4r_+P .
\end{align}
\end{lemma}

\begin{proof}
\begin{multline}
\label{eq:estimate_proof}
    0 \le \int_0^1r\big(|\phi| + 2v|\phi|'\big)^2\dd v = \int_0^1r|\phi|^2\dd v + 4\int_0^1vr|\phi||\phi|'\dd v + 4\int_0^1v^2r(|\phi|')^2\dd v \\
    = \int_0^1r|\phi|^2\dd v - 2\int_0^1(vr' + r)|\phi|^2\dd v + 4r_+P \le 4r_+P - \int_0^1r|\phi|^2\dd v ,
\end{multline}
where, in going to the second line, I used integration by parts. The boundary term vanishes because $\phi(1) = 0$. The last inequality uses that $r' \ge 0$ everywhere. 
\end{proof}

\begin{remark}
The bound obtained in Lemma~\ref{lemma:Hardy} is not sharp: in order to saturate the first inequality in Eq.~\eqref{eq:estimate_proof}, $|\phi| \propto v^{-1/2}$, while in order to saturate the last inequality, $r' = 0$ which only holds if $\phi \equiv 0$. Since no solution has $\phi \equiv 0$, the inequality is never saturated. 
\end{remark}

\subsection{The main estimates}

The following identity will be the starting point: 
\begin{proposition}[The splitting identity]
\label{prop:exact_identity}
The charge $Q_+$ can be written as (cf. Eq.~\eqref{eq:Maxwell})
\begin{align}
    Q_+= \mathfrak{e}\int_0^1vr^2\operatorname{Im}(\phi\overline{\phi'})\dd v + \int_0^1Q\dd v .
\end{align}
\end{proposition}

\begin{proof}
The identity follows from simple integration by parts:
\begin{align}
    \mathfrak{e}\int_0^1vr^2\operatorname{Im}(\phi\overline{\phi'})\dd v = \int_0^1vQ'\dd v = (vQ)\big|_0^1 - \int_0^1Q\dd v = Q_+ - \int_0^1Q\dd v .
\end{align}
\end{proof}

The next two propositions will be bounding the right-hand side, starting with the first term: 
\begin{proposition}[The charge production bound]
\label{prop:weighted_charge_production_bound}
The following estimates hold for the charge and the weighted charge production term:
\begin{align}
    |Q_+| \le 2|\mathfrak{e}|r_+^2\sqrt{\Big(\frac{r_0}{r_+}E - P\Big)P} , \\
    \label{eq:weighted_charge_production}
    \Big|\int_0^1vr^2\operatorname{Im}(\phi\overline{\phi'})\dd v\Big| \le 2r_+^2\sqrt{(H - P)P} .
\end{align}
\end{proposition}

\begin{proof}
For the first estimate, split $r^2|\operatorname{Im}(\phi\overline{\phi'})| = (\sqrt{r(|\phi'|^2 - (|\phi|')^2)})(r^{3/2}|\phi|)$ and use Cauchy--Schwarz to find
\begin{multline}
    Q_+^2 = \mathfrak{e}^2\Big(\int_0^1r^2\operatorname{Im}(\phi\overline{\phi'})\dd v\Big)^2 \le \mathfrak{e}^2\Big(\int_0^1r(|\phi'|^2 - (|\phi|')^2)\dd v\Big)\Big(\int_0^1r^3|\phi|^2\dd v\Big) \\
    \le 4\mathfrak{e}^2r_+^3(r_0E - r_+P)P ,
\end{multline}
using $v^2 \le 1$ so that $\int_0^1r(|\phi|')^2\dd v \ge r_+P$ on the first factor, and using $r^3 \le r_+^2r$ together with Lemma~\ref{lemma:Hardy} on the second factor. The second estimate is obtained in an analogous way, by splitting $vr^2|\operatorname{Im}(\phi\overline{\phi'})| = (v\sqrt{r(|\phi'|^2 - (|\phi|')^2)})(r^{3/2}|\phi|)$ and using Cauchy--Schwarz to find
\begin{multline}
    \Big(\int_0^1vr^2\operatorname{Im}(\phi\overline{\phi'})\dd v\Big)^2 \le \Big(\int_0^1v^2r(|\phi'|^2 - (|\phi|')^2)\dd v\Big)\Big(\int_0^1r^3|\phi|^2\dd v\Big) \\
    \le 4r_+^4(H - P)P ,
\end{multline}
again using $r^3 \le r_+^2r$ and Lemma~\ref{lemma:Hardy} on the second factor. 
\end{proof}

\begin{remark}
\label{remark:lossiest_step}
The split $|\operatorname{Im}(\phi\overline{\phi'})| = |\phi|\sqrt{|\phi'|^2 - (|\phi|')^2}$ is exact, and follows by decomposing $\phi = |\phi|e^{i\theta}$, so that $\operatorname{Im}(\phi\overline{\phi'}) = -|\phi|^2\theta'$. The Cauchy--Schwarz estimate used to prove the weighted charge production bound~\eqref{eq:weighted_charge_production} in Proposition~\ref{prop:weighted_charge_production_bound} is therefore sharp when $v\theta' = r$. This is precisely why the solution constructed in Theorem~\ref{thm:subextremal} has $\theta \propto \log v$. The lossiest step is the application of Lemma~\ref{lemma:Hardy}, which is never sharp. 
\end{remark}

The following proposition bounds the bulk charge term: 
\begin{proposition}[The bulk charge bound]
\label{prop:bulk_charge_bound}
The constraint implies the following bound for the bulk charge: 
\begin{align}
    \Big|\int_0^1Q\dd v\Big| < r_+\sqrt{\Big(1 + f_0E^{-1}\Big)\Big(1 - \frac{1}{3}\Big(2 + \frac{r_0}{r_+}\Big)H\Big)} .
\end{align}
\end{proposition}

\begin{proof}
Split $Q = (Q/r)(r)$, apply Cauchy--Schwarz, and use the constraint~\eqref{eq:constraint} and corollary~\ref{cor:concave} to find
\begin{align}
    \Big(\int_0^1Q\dd v\Big)^2 \le \Big(\int_0^1\frac{Q^2}{r^2}\dd v\Big)\Big(\int_0^1r^2\dd v\Big) < (1 + f_0E^{-1})r_+^2\Big(1 - \frac{1}{3}\Big(2 + \frac{r_0}{r_+}\Big)H\Big) ,
\end{align}
using that $\frac{r_0}{r'(0)} = E^{-1}$ (cf. Proposition~\ref{prop:elementary}). 
\end{proof}

\subsection{The bound}

\begin{theorem}[Theorem~\ref{mainthm:extremal}, necessary part]
\label{thm:extremal}
In the extremal case $|Q_+| = r_+$, a necessary condition for a solution (cf. Definition~\ref{def:solution}) to exist is that
\begin{align}
    |\mathfrak{e}|r_+ > \frac{1}{3} .
\end{align}
\end{theorem}

\begin{proof}
As outlined at the beginning of this section, the idea is to start from Proposition~\ref{prop:exact_identity}, and bound the weighted charge production term using Proposition~\ref{prop:weighted_charge_production_bound} and the bulk charge term using Proposition~\ref{prop:bulk_charge_bound}: 
\begin{multline}
    r_+ = \Big|\mathfrak{e}\int_0^1vr^2\operatorname{Im}(\phi\overline{\phi'})\dd v + \int_0^1Q\dd v\Big| \\
    < 2|\mathfrak{e}|r_+^2\sqrt{(H - P)P} + r_+\sqrt{\Big(1 + f_0E^{-1}\Big)\Big(1 - \frac{1}{3}\Big(2 + \frac{r_0}{r_+}\Big)H\Big)} .
\end{multline}
From Proposition~\ref{prop:weighted_charge_production_bound} the following upper bound on $E^{-1}$ is obtained:
\begin{align}
    E^{-1} \le 4\mathfrak{e}^2r_+^2\frac{r_0}{r_+}P .
\end{align}
Combining these two estimates, and using that $f_0 \le 1$,
\begin{align}
\label{eq:step}
    1 < 2|\mathfrak{e}|r_+\sqrt{(H - P)P} + \sqrt{\Big(1 + 4\mathfrak{e}^2r_+^2\frac{r_0}{r_+}P\Big)\Big(1 - \frac{1}{3}\Big(2 + \frac{r_0}{r_+}\Big)H\Big)} .
\end{align}
The remaining steps are pure algebra. Suppose $2|\mathfrak{e}|r_+\sqrt{(H - P)P} < 1$ because otherwise $(H - P)P \le \frac{1}{4}H^2$ and $H < 1$ imply that $|\mathfrak{e}|r_+ > 1 > \frac{1}{3}$. Then
\begin{align}
\label{eq:constant_cancels}
    \Big(1 - 2|\mathfrak{e}|r_+\sqrt{(H - P)P}\Big)^2 < \Big(1 + 4\mathfrak{e}^2r_+^2\frac{r_0}{r_+}P\Big)\Big(1 - \frac{1}{3}\Big(2 + \frac{r_0}{r_+}\Big)H\Big) .
\end{align}
When expanded, both the left- and the right-hand side have the constant term $1$. The remaining terms may be divided by $H > 0$, so that
\begin{multline}
\label{eq:constant_cancelled}
    0 \le \Big(4(H - P) + \frac{4}{3}\Big(2 + \frac{r_0}{r_+}\Big)\frac{r_0}{r_+}H\Big)\mathfrak{e}^2r_+^2\frac{P}{H} \\
    < 4\Big[\sqrt{\Big(1 - \frac{P}{H}\Big)\frac{P}{H}} + |\mathfrak{e}|r_+\frac{r_0}{r_+}\frac{P}{H}\Big]|\mathfrak{e}|r_+ - \frac{1}{3}\Big(2 + \frac{r_0}{r_+}\Big) .
\end{multline}
Define $-1 \le \chi \le 1$ by $\frac{P}{H} = \frac{1+\chi}{2}$. Then
\begin{multline}
    0 < 2\Big(\sqrt{1 - \chi^2} + |\mathfrak{e}|r_+\frac{r_0}{r_+}\chi\Big)|\mathfrak{e}|r_+ + 2\mathfrak{e}^2r_+^2\frac{r_0}{r_+} - \frac{1}{3}\Big(2 + \frac{r_0}{r_+}\Big) \\
    \le 2|\mathfrak{e}|r_+\sqrt{1 + \frac{r_0^2}{r_+^2}\mathfrak{e}^2r_+^2} + 2\frac{r_0}{r_+}\mathfrak{e}^2r_+^2 - \frac{1}{3}\Big(2 + \frac{r_0}{r_+}\Big)
\end{multline}
where in the last step I used that for $\xi > 0$, $\sqrt{1 - \chi^2} + \xi\chi \le \sqrt{1 + \xi^2}$. Assume now that $2\frac{r_0}{r_+}\mathfrak{e}^2r_+^2 \le \frac{1}{3}(2 + \frac{r_0}{r_+})$ because otherwise $|\mathfrak{e}|r_+ > \sqrt{\frac{r_0}{r_+}\mathfrak{e}^2r_+^2} > \sqrt{\frac{1}{6}(2 + \frac{r_0}{r_+})} > \frac{1}{3}$. It follows that
\begin{align}
    \Big(\frac{1}{3}\Big(2 + \frac{r_0}{r_+}\Big) - 2\frac{r_0}{r_+}\mathfrak{e}^2r_+^2\Big)^2 < 4\mathfrak{e}^2r_+^2\Big(1 + \frac{r_0^2}{r_+^2}\mathfrak{e}^2r_+^2\Big) , \\
    \implies \Big(\frac{1}{3}\Big(2 + \frac{r_0}{r_+}\Big)\Big)^2 < 4\mathfrak{e}^2r_+^2\Big(1 + \frac{1}{3}\frac{r_0}{r_+}\Big(2 + \frac{r_0}{r_+}\Big)\Big) ,
\end{align}
where conveniently, the $\mathfrak{e}^4r_+^4$ terms cancelled. Finally, it follows that
\begin{align}
\label{eq:r0_dependent_bound}
    |\mathfrak{e}|r_+ > \frac{1}{3}\Big(1 + \frac{1}{2}\frac{r_0}{r_+}\Big)\Big(1 + \frac{1}{3}\frac{r_0}{r_+}\Big(2 + \frac{r_0}{r_+}\Big)\Big)^{-1/2} > \frac{1}{3} .
\end{align}
\end{proof}

\begin{remark}
The step in the proof of Theorem~\ref{thm:extremal} that fails in the subextremal case occurs at going from Eq.~\eqref{eq:constant_cancels} to Eq.~\eqref{eq:constant_cancelled}, where the constant terms could be cancelled so that the resulting simplified expression could be divided by $H$. In the subextremal case, instead of zero, the left-hand side of Eq.~\eqref{eq:constant_cancelled} would be $(Q_+^2/r_+^2 - 1)H^{-1}$, which is negative and unbounded in the subextremal case. Letting $H \to 0$ is what allows for the subextremal solutions of Theorem~\ref{thm:subextremal} to acquire enough charge. 
\end{remark}

\begin{remark}
The bound proven in the proof of Theorem~\ref{thm:extremal} is slightly stronger for large $r_0/r_+$. The bound in Eq.~\eqref{eq:r0_dependent_bound} approaches $|\mathfrak{e}|r_+ > \tfrac{1}{3}$ as $r_0/r_+ \to 0$, and $|\mathfrak{e}|r_+ > \tfrac{1}{2\sqrt{2}}$ as $r_0/r_+ \to 1$. 
\end{remark}

\begin{remark}
\label{remark:unused_assumptions}
Not all assumptions of Definition~\ref{def:solution} were used in the proof of Theorem~\ref{thm:extremal}. The endpoint condition $\phi(0) = 0$ was never used, and only the global constraint
\begin{align}
    \int_0^1\frac{Q^2}{r^2}\dd v < 1 + f_0E^{-1} ,
\end{align}
was used, rather than the (stronger) pointwise condition~\eqref{eq:constraint}. Presumably, obtaining a sharp bound will require the incorporation of Eq.~\eqref{eq:constraint}. 
\end{remark}

\begin{remark}
\label{remark:regular_centre}
Theorem~\ref{thm:extremal} also holds when replacing the gluing conditions at $v = 0$ with a regular centre $r_0 = 0$. The constraint~\eqref{eq:constraint} becomes $\int_0^v\frac{Q^2}{r^2}\dd x < v$ for $v \in (0,1]$. Because Lemma~\ref{lemma:Hardy} never used the boundary conditions at $v = 0$, Proposition~\ref{prop:weighted_charge_production_bound}'s estimate~\eqref{eq:weighted_charge_production} still holds. With the new constraint, Proposition~\ref{prop:bulk_charge_bound} holds with $r_0 = E^{-1} = 0$. It follows that Eq.~\eqref{eq:step} holds with $r_0 = 0$, and from this point onward the proof is the same. 
\end{remark}

\section{Concluding remarks}

\subsection{The status of the bound}

The bound obtained in Theorem~\ref{thm:extremal} together with the solutions obtained in \S\ref{sec:failure_at_extremality} constrain the sharp bound to lie in the interval $[\tfrac{1}{3}, \sqrt{\frac{3 + \sqrt{33}}{3}}]$. The gap is fairly modest, being about a factor $5$. The lossiest step is the application of Lemma~\ref{lemma:Hardy} (cf. remark~\ref{remark:lossiest_step}), so any attempt to improve the bound by trying to sharpen the proof presented here will have to start here. Obtaining a sharp bound will probably require a new insight rather than just a sharpening of the argument presented here, because not all assumptions laid out by Definition~\ref{def:solution} were used in the proof. Most notably, a weaker version of the pointwise constraint~\eqref{eq:constraint} was used (cf. remark~\ref{remark:unused_assumptions}). 

\subsection{Future directions}

As mentioned in the introduction, the equations of motion for the metric Ansatz used in \cite{CGRS2026} to numerically construct vacuum extremal collapse in five space-time dimensions are tantalizingly similar to the Einstein--Maxwell--charged scalar field system. It would be interesting to see if the techniques used in \S\ref{sec:subextremal} can be used to analytically construct counterexamples to the third law in vacuum five-dimensional general relativity. 

Besides obtaining a tighter bound I have also left the question of the existence of higher regularity solutions open. I strongly suspect that Theorem~\ref{thm:subextremal} also holds for $C^k$ gluing solutions (cf. remark~\ref{remark:C1_modification}), but I did not try to prove this.

\section{Acknowledgements}

\noindent I would like to express my gratitude to Jorge Santos for a stimulating discussion, which motivated me to work on this problem again after an earlier attempt had failed. I would also like to thank Maxime Gadioux for pointing out some minor errors in an earlier version of the proof of Theorem~\ref{thm:extremal}. This work was supported by the Natural Sciences and Engineering Research Council of Canada.

\bibliography{references}

\end{document}